\documentclass[9pt]{amsart}
\usepackage[english]{babel}
\usepackage{amsmath,amsthm}
\usepackage{amsfonts}
\newtheorem{thm}{Theorem}[section]
\newtheorem{cor}[thm]{Corollary}
\newtheorem{lem}[thm]{Lemma}
\newtheorem{prop}[thm]{Proposition}
\newtheorem{rem}[thm]{Remark}
\theoremstyle{definition}

\numberwithin{equation}{section}

\newcommand{\Real}{\mathbb R}

\newcommand{\spectrum}{{\sigma\,}}

\newcommand\Complex{\mathbb C}
\newcommand{\bbP}{\mathbb P}
\newcommand{\bbQ}{\mathbb Q}

\newcommand{\bbQP}{{\bbQ\times\bbP}}
\newcommand{\bbQQ}{{\bbQ\times\bbQ}}
\newcommand{\bbPP}{{\bbP\times\bbP}}

\newcommand{\LQP}{L_{\bbQ\times\bbP}}
\newcommand{\bbt}{{2}}
\newcommand{\Lt}{{L^\bbt}}
\newcommand{\LPt}{{L_\bbP^\bbt}}
\newcommand{\LQt}{{L_\bbQ^\bbt}}
\newcommand{\LQPt}{{L_{\bbQ\times\bbP}^\bbt}}

\newcommand{\bbU}{\mathcal{U}}
\newcommand{\bbV}{\mathcal{V}}

\newcommand{\bbq}{\mathfrak{q}}

\newcommand{\Lq}{L^\bbq}

\newcommand{\bbg}{\mathfrak{g}}
\newcommand{\bbf}{\mathfrak{f}}
\newcommand{\bbh}{\mathfrak{h}}

\newcommand{\cP}{\mathcal{P}}
\newcommand{\cQ}{\mathcal{Q}}
\newcommand{\cT}{\mathcal{T}}
\newcommand{\cS}{\mathcal{S}}

\newcommand{\cK}{\mathcal{K}}
\newcommand{\cD}{\mathcal{D}}
\newcommand{\cA}{\mathcal{A}}

\newcommand{\cH}{\mathcal{H}}
\newcommand{\cC}{\mathcal{C}}

\newcommand{\cTs}{{\cT^\dag}}

\newcommand\scpr[2]{\langle #1, #2 \rangle}

\newcommand\sinc{{\rm sinc}}

\title[Functional Analysis in Algorithms]{HMC, an  Algorithms in Data Mining, the Functional Analysis approach.}%
\author{Soumyadip Ghosh, Yingdong Lu, Tomasz Nowicki}
\address{  IBM T.J.~Watson Research Center, 1101 Kitchawan Road, Yorktown Heights, NY 10598, US}
\email{ghoshs, yingdong, tnowicki@us.ibm.com}
\begin{document}
\begin{abstract}
    The main purpose of this paper  is to facilitate the communication between the Analytic, Probabilistic and Algorithmic communities.
  We present a proof of convergence of the Hamiltonian (Hybrid) Monte Carlo algorithm from the point of view of the
  Dynamical Systems, where the evolving objects are densities of probability distributions and the tool are derived from the  Functional Analysis.
\end{abstract}
\maketitle
\section{Introduction}
\subsection*{Functional Analysis for all Functioning Algorithms}
We observed that all too often not only we do not speak a common language but also we do not see the reason to communicate and to see the problems through a different eye.
We tried, on the example an HMC algorithm, to gently (relatively speaking) build the bridges and make the "other" methods clearer and more comprehensible.
As the paper is addressed to people not necessarily fluent in Functional Analysis we make the proofs quite extended at some places and sometimes a bit hand-waving.
It is difficult to strike the balance between clarity and rigor.

\subsection*{Algorithms as source of inspiration}
Recent development and usage of Machine Learning (ML) Data Mining (DM) and Artificial Intelligence (AI)
resulted in a vast variety of new or refurbished algorithms to deal with large data sets, be it collected or streamed. Such new methods are usually tested on some data sets, however not very often they are thoroughly vetted by theoretical means. Algorithms tend to rely on discrete models, but the nature of data and approximation approaches suggest rather a continuous point of view. In our opinion one should go even further, the right objects of investigation of algorithms should
not be the continuous parameters but rather very general features of data such as distributions.

We perceive the algorithms as iterative transformations of the points in some underlying domains.
The leading idea is to move from a relatively simple objects such as finite sets or points in finite dimensional Euclidean spaces with quite complicated  transformations to
simple transformations in richer spaces such as distributions in functional spaces.

\subsection*{Hamiltonian Monte Carlo}
Or Hybrid Monte Carlo (HMC) algorithm is a method to obtain random samples from a
(target) probability distribution $\bbf/\int_\bbQ\bbf$ on the space $\bbQ$ whose density is known only up to a factor,
that is to say that $\bbf$ is known, but $\int_\bbQ\bbf$ is not, or at least is very difficult to calculate.
It is an algorithms known for a while~\cite{1} of the Metropolis-Hastings type used to estimate the integrals.
There are known proofs of convergence~\cite{2}.
Our goal is to provide a clear and understandable reason why HMC algorithm converges to the right limit for densities in the spaces $\Lt(\bbQ)$.
We refer to our papers~\cite{3,4,5} for other approaches to HMC, probabilistic, algorithmic and analytic in $\Lq$.
Here we concentrate on the convergence in $\Lt$.

HMC performs by iterating the following steps. Given an initial distribution $h$ (sample points) in a given space $\bbQ$ double (the dimension of) the space by considering $\bbQP$, with $\bbP\sim\bbQ$. Then spread each point $q\in\bbQ$ to a point $(q,p)\in\bbQP$
by sampling $p\in\bbP$ from a distribution of choice $\bbg$, where $\bbg>0$ on $\bbP$ and $\int_\bbP\bbg=1$. Then move each point $(q,p)$ to a new point $(Q,P)=H(q,p)$, where the transformation $H:\bbQP\to\bbQP$ satisfies some special invariance properties with respect to $\bbf$ and $\bbg$ and the underlying base measure on~$\bbQP$. Finally project $(Q,P)$ on $\bbQ$ providing a new sample of points $Q$ in $\bbQ$ with a new distribution $\hat{h}$ which shall be used as the initial sample (or distribution) for the next step.
With the right choice of $H$ the iteration of this procedure will result in an approximate sample from the target distribution.

\subsection*{Moving to functional spaces}
The success of the algorithm lies in the appropriate choice of the transformation $H$.
The Hamiltonian part of the algorithm's name is due to the Hamiltonian motion $H$. It turns out that  if the target distribution has a density proportional to a given function $\bbf$ and the distribution of choice is $\bbg$ then $H$ is a Hamiltonian motion generated by the Hamiltonian energy $\cH(q,p)=-\log(\bbf(q)\cdot\bbg(p))$. That is if $(Q,P)=H(q,p)$ is the solution of the time evolution
$\dot{Q}=\partial\cH/\partial P$, $\dot{P}=-\partial\cH/\partial Q$ after time $t$ with initial point $(q,p)$ then $H$ has the needed invariance properties for HMC to converge to a distribution proportional to $\bbf$. Effectively it means that one can obtain a normalizing constant $\int_\bbQ \bbf$ or any expected value of a function $\phi$ with respect to the distribution proportional to $\bbf$: $\int\phi\cdot\bbf/\int\bbf$.

In terms of the densities of the involved distributions one can present HMC as follows: Given some initial distribution $h(q)$ on $\bbQ$ one produces a joint distribution $h\cdot \bbg$ by $(h\cdot\bbg)(q,p)=h(q)\cdot \bbg(p)$ on $\bbQP$ then moves the points $(q,p)\mapsto (Q,P)=H(q,p)$ producing another distribution $(h\cdot\bbg)\circ H(q,p)=h(Q)\cdot(P)$ in $\bbQP$ and finally projects the last one onto $\bbQ$
by calculating the marginal $\int_\bbP (h\cdot\bbg)\circ H(q,p)\,dp$, which is a result of the action of the algorithm in one step. In short
\begin{equation}
  \label{eqdef:cT}
  \cT(h)(q)=\int_\bbP (h\cdot \bbg)\circ H (q,p)\,dp
\end{equation}
  and from a rather complicated algorithm we receive a relatively simple, linear operator on some space of integrable functions. The convergence of the algorithm  corresponds to the convergence of sequences of iterates of $\cT$.

\begin{rem}\
\begin{enumerate}
\item  The distribution $\bbg$ may depend on the point $q$, $\bbg(p|q)$ as long as for (almost) all $q$ it satisfies the required conditions.
\item It is clear that the Hamiltonian motion $H$ does not depend on the constant normalizing factor in front of $\bbf$
\item The motion $H$ in practical implementation is performed by the \emph{leap-frog algorithm} which displays the needed invariance properties. We shall not deal with it in this paper.
\item An example of the situation where the target distribution is known up to the normalizing constant is the Bayesian update. In order to establish the distribution of (random)
    parameters $\theta$ influencing the outcome $D$ of the observations, when we know all the probabilities $P(D|\theta)$ one uses the knowledge of the outcome $D$ to improve the estimate:
    Given the estimate distribution $\pi(\theta)$ before the experiment we calculate $\hat{\pi}(\theta)=
    P_\pi(\theta|D)=P_\pi(\theta, D)/P_\pi(D)=P(D|\theta)\cdot\pi(\theta)/\sum_{\theta'} P(D|\theta')\cdot\pi(\theta')$ and take the distribution $\hat{\pi}$ as a new estimate. However the sum (integral) in the denominator may be not that easy to calculate. This yields to $\hat{\pi}(\theta)\sim P(D|\theta)\pi(\theta)$ without the normalizing factor.
\end{enumerate}
\end{rem}

\section{Results}
\subsection*{Convergence under invariance and coverage properties}
Assume that the motion $H:\bbQP\to\bbQP$ (measurable spaces with measures $dq$ and $dp$) satisfies the following \emph{invariance and coverage properties} when \\
Given $0\le\bbf:\bbQ\to\Real$\quad $\int_\bbQ\bbf<\infty$, $0\le \bbg:\bbP\to\Real$\quad $\int_\bbP\bbg=1$:
\begin{eqnarray}
(\bbf\cdot\bbg)\circ H=\bbf\cdot\bbg &&\label{ass:inv fg}\\
 \iint_\bbQP A\circ H =\iint_\bbQP A && \text{ for any integrable } A \label{ass:inv A}\\
 Q(q,\bbP)=\bbQ&&\text{ for (almost) every }q\label{ass:Q onto Q}
\end{eqnarray}
The space $\bbQ$ may be restricted to equal the support of $\bbf$.

Define  $\cT h=\int_\bbP (h\cdot\bbg)\circ H$  as in \eqref{eqdef:cT} and $\cT^{n+1}=\cT^n\circ\cT$.
 The \emph{adjoint} operator $\cTs$ is given by the same formula with $H^{-1}$ in place of $H$ and is described below in Section~\ref{sec:adjoint}.
 A \emph{self-adjoint} operator satisfies $\cTs=\cT$, see~\eqref{eqn:selfadjoint}.

Let $\Lt$ denotes the space of square-integrable functions $h:\bbQ\to\Real$ such that $||h||_\bbt^\bbt=\int_\bbQ |h|^\bbt/\bbf<\infty$ and the support of $h$ is included in the support of $\bbf$ (\emph{i.e.}~$\bbQ$).
\begin{thm}\label{thm:strong conv}
   Under the above invariance and coverage conditions and when the operator $\cT$ is self-adjoint then for every $h\in\Lt(\bbQ)$ the sequence $\cT^n h$ converges strongly in $\Lt$ to
   $\bbf \cdot \int h/\int\bbf$. The direction of $\bbf$ is the unique direction of fixed points. Except of the eigenvalue 1 with multiplicity 1,  all the spectrum is contained in the interior of the unit disc.
 \end{thm}
\begin{rem}
\
\begin{enumerate}
\item The Hamiltonian motion satisfies two first integral invariance assumptions, as both the Hamiltonian (energy function equal here $-\log(\bbf(q)\cdot\bbg(p))$) and the Lebesgue measure are invariant under such a motion.
\item
    The coverage  property $Q(q,\bbP)=\bbQ$ can be weakened to a statement of an eventual coverage, not necessarily in one step. Some type of irreducibility must be assumed to avoid complete disjoint domains of the motion and hence an obvious non existence of a (unique) limit.
\item The support condition takes care of some initialization problems with the division by 0. This can be formally avoided by working in the space of likelihoods, see below.
\item
    The self-adjointness condition is not very restrictive, as one can always use the composition $\cTs\circ \cT $ which is self-adjoint and satisfy all the needed properties. Also in case of any even auxiliary distribution $\bbg(p)=\bbg(-p)$ on $\Real^d$, such as standard Gaussian,  the operator is always self-adjoint, see Lemma~\ref{lem:sigma}.
\end{enumerate}
\end{rem}

\subsection*{Exponential convergence under uniformly strong logarithmic concavity}
We say that $\bbh:\Real^d\to\Real$ is \emph{uniformly strongly logarithmic concave} if for almost every $Q$ the (symmetric)  Hessian $-\partial^2\log(\bbh(Q))/\partial Q^2$
 has its spectrum contained in some finite, positive interval $\{z: 0<\lambda\le \Lambda<\infty\}$ independent on $Q$.
 Gaussian auxiliary distributions are obviously uniformly strongly concave.
 Speaking informally one may say that the density lies between two Gaussians.

 For the next Theorem we need a \emph{stronger coverage condition}. We say that the motion $H:(q,p)\mapsto(Q,P)$ is fully invertible if given fixed values of (almost) any two of the four  variables $(q,p,Q,P)$ the other two are connected by a differentiable bijection. In particular given $q$ the derivative $\partial Q/\partial p$ of the map $p\mapsto\cQ_q(p)=Q$
 is smoothly invertible and the same holds for  $q\mapsto\cP_p(q)=P$.
\begin{thm}
  \label{thm:Hamil  geom conv}
  Suppose that the target distribution $\bbf:\Real^d\to\Real$ and the auxiliary distribution $\bbg:\Real^d\to\Real$ are both uniformly strictly logarithmic concave.
  Let $H_t$ be the Hamiltonian motion defined by the Hamiltonian $\cH(Q,P)=-\log(\bbf(Q)\cdot\bbg(P))$ and $\cT$ the operator defined by $H_t$. Let assume it is self-adjoint.
  Then for $t>0$  small enough
  the iterations of the operator   $\cT_t$  are converging geometrically to the map $h\mapsto \int h\cdot \bbf/\int \bbf$.
  \[
  \exists(0<\rho<1)\,\forall(h\in\Lt(\bbQ))\,\forall(n) \quad ||\cT^n h-\bbf\frac{\int h}{\int\bbf}||_2\le \rho^n ||h||_2\,.
  \]
\end{thm}
One can prove that the uniformly strong logarithmic concavity assumption is needed only outside an arbitrary bounded region, see~\cite{4}.

\subsection*{The proofs are not very hard}
For Theorem~\ref{thm:strong conv} we observe that $\cT$ is in fact~\eqref{lem:cT prop alt} an averaging map, thus by the convexity of $x\mapsto x^\bbt$ the norm of $h$ decreases~\eqref{lem:cT prop norm} under $\cT$, sharply unless (by coverage assumption) $h=\alpha\bbf$. The space $\Lt$ is reflexive, thus  bounded sequences have weak accumulation points. Using self-adjointness defined in Section~\ref{sec:adjoint} we prove that each accumulation point must be of form $\alpha\bbf$, proving weak convergence. Meanwhile the proof of the convergence of the norms provides (Proposition~\ref{prop:norm of hinf}) strong convergence. The spectral properties follow from Remark~\ref{rem:cT spec}.

The proof of Theorem~\ref{thm:Hamil  geom conv}  relies on the representation of the operator $\cT$ as a kernel operator~\eqref{eqn:ct kernel} and the proof in Subsection~\ref{subsec:Proof Thm geom conv}
(somewhat lengthy in calculation but not too deep) that the $\Lt$ norm of that kernel~\eqref{eqn: norm of K} is finite,
hence the kernel is compact and the operator $\cT$ has the spectral gap which provides the geometric (or exponential) rate of convergence.

\section{The operator $\cT$ in $\Lt$}\label{sec: operator cT}
\subsection*{The Hilbert space $\Lt$}
For $h\in\Lt$ we have a standard norm $||h||_\bbt^\bbt=\int_\bbQ \left(\frac{h}{\bbf}\right)^\bbt \bbf$
and a scalar  product $\scpr{\cdot}{\cdot}:\Lt\times\Lt\to\Real:\scpr{a}{b}=\int_\bbQ\frac{a\cdot b}{\bbf}$.

We shall assume $h\ge 0$ unless stated otherwise.
Call $h/\bbf$ a \emph{likelihood} (up to an irrelevant  normalizing constant $\int\bbf$) of $h$ with respect to~$\bbf$.   The space of  $\tilde{\Lt}=\{\tilde{h}:\int_\bbQ |\tilde{h}|^\bbt \bbf<\infty\}$ of likelihoods $\tilde{h}=h/\bbf$ is isometric to $\Lt$.
\begin{lem}\label{lem:prop of nbc}
\begin{eqnarray}
\label{lem:cT prop sc}
a, b\in \Lt\Rightarrow\quad
\scpr{a}{b}&\le& ||a||_\bbt\cdot||b||_\bbt,\\
\label{lem:cT prop f}
\text{we have}\quad \bbf\in\Lt\text{ with }\quad ||\bbf||_\bbt^\bbt&=&\int_\bbQ \bbf,
\\
\label{lem:cT prop int f}
h\in\Lt\Rightarrow\quad\scpr{h}{\bbf}&=&\int_\bbQ h\quad
\,.
\end{eqnarray}
\end{lem}
\begin{proof}
 \eqref{lem:cT prop sc} is the H\"older inequality for likelihoods $a/\bbf$ and $b/\bbf$ in the space $\tilde{\Lt}$.\\
\eqref{lem:cT prop f} and \eqref{lem:cT prop int f} follow directly from the definitions.
\end{proof}
\subsection*{Operator $\cT$}
\begin{lem}[Properties of $\cT$]\label{lem:cT prop}
For $0\le h\in\Lt$:
\begin{eqnarray}
\label{lem:cT prop alt}
\cT h &=&  \bbf\cdot\int_\bbP \frac{h}{\bbf}\circ H \cdot\bbg
\\
\label{lem:cT prop int}
\int_\bbQ\cT h &=&  \int_\bbQ h\\
\label{lem:cT prop norm}
||\cT h||_\bbt &\le& ||h||_\bbt \\
\label{lem:cT prop eq}
\text{The equality in \eqref{lem:cT prop norm} occurs}&\text{iff}& h=\alpha\cdot\bbf \text{ a.e.}\,,
\end{eqnarray}
where $\alpha=\alpha(h)=\int h/\int \bbf$.
\end{lem}
Equation~\eqref{lem:cT prop alt} gives an explicit formula for $\tilde{\cT}(\tilde{h})$ acting in $\tilde{L}^2$.
\begin{proof}\ \\
\eqref{lem:cT prop alt}:
Using the invariance properties we have
$\int_\bbP \frac{h}{\bbf}\circ H\cdot (\bbf\cdot\bbg)\circ H=
 \int_\bbP \frac{h}{\bbf}\circ H \cdot (\bbf\cdot\bbg)
$
and $\bbf$ does not depend on $p\in\bbP$.\\
\eqref{lem:cT prop int}:
$\int_{\bbQ}\int_{\bbP} (h\cdot\bbg)\circ H =\iint_{\bbQP}(h\cdot \bbg)=
  \left(\int_{\bbQ} h\right)\left(\int_{\bbP}\bbg\right)$.\\
\eqref{lem:cT prop norm}:
$
||\cT h||_\bbt^\bbt=\int_\bbQ \left(\int_\bbP\frac{h}{\bbf}\circ H\cdot\bbg\right)^{\bbt}\bbf
\le
\int_\bbQ \int_\bbP\left(\frac{h}{\bbf}\circ H\right)^{\bbt}\bbg\,\bbf
 =
\iint_\bbQP\left(\frac{h}{\bbf}\right)^\bbt\circ H\cdot(\bbg\cdot\bbf)\circ H
 =
\iint_\bbQP\left(\frac{h}{\bbf}\right)^\bbt\bbg\cdot\bbf
$ the last one being equal to
$\left(\int_\bbQ\left(\frac{h}{\bbf}\right)^\bbt\bbf\right)\cdot\left(\int_\bbP\bbg\right)=
||h||_\bbt^\bbt$. For a given $q$ the equality occurs only if $(h/\bbf)(H(q,p))$ is a constant for ($\bbg$-)almost all $p$, but the coverage assumption assures that the constant $(h/\bbf)\circ H(q,\bbP)$ is the same for $\bbf$-almost all $q\in\bbQ$: $(h/\bbf)(\bbQ)=\int_\bbQ h/\int_\bbQ\bbf$, the value follows from~\eqref{lem:cT prop int}.
Note that~\eqref{lem:cT prop norm} hides the formula for the variance with respect to probability $\bbg$, the random variable being  the transported likelihood.
\end{proof}
\begin{rem}
  \label{rem:cT spec}
The operator $\cT$ is an averaging operator of the (transported) likelihood $h/\bbf$ with respect to the probability $\bbg$. The scalar product is monotone:
$0\le a\le b,\quad 0\le c\le d$ implies $\scpr{a}{c}\le \scpr{b}{d}$ and
$\cT$ is positive, in particular if $a\le b$ then $\cT a\le \cT b$.
The function  $\bbf$ provides the eigendirection of fixed points and by \eqref{lem:cT prop norm}
$\cT$ has its spectrum in the unit disk,
with 1 being a unique eigenvalue on the unit circle and has multiplicity 1.
For any $h\in \Lt$ one has the unique decomposition $h=\alpha\bbf+(h-\alpha\bbf)$ where $\alpha\bbf$ is a direction of the fixed points and
$h-\alpha\bbf\in N=\{a\in\Lt:\int a =0\}$ lies in an invariant subspace.
It is not \emph{a priori} clear under what conditions the eigen-value $1$ is isolated in the spectrum, in other words wether the contraction $||\cT h||<||h||$ is uniform on $N$, which would imply $\cT^n N\to \{0\}$ (point-wise) with exponential speed.
\end{rem}

\section{The adjoint operator $\cTs$}\label{sec:adjoint}
As $H$ is invertible the inverse map $H^{-1}$ is well defined as enjoys the same invariance properties as $H$. It turns out that
the operator ${\cTs}$ defined by $H^{-1}$:
\begin{equation}\label{eqn:selfadjoint}
{\cTs}h=\int_\bbP (h\cdot\bbg)\circ H^{-1}
\end{equation}
is \emph{adjoint} to $\cT$ with respect to the duality functional $\scpr{\cdot}{\cdot}$,
namely
\begin{lem}
  For $h,k\in\Lt$:
  \begin{equation}\label{lem: cTs}
    \scpr{\cT h}{k}=\scpr{h}{{\cTs} k}
  \end{equation}
\end{lem}
\begin{proof}
  Using \eqref{lem:cT prop alt} and invariance
$\scpr{\cT h}{k}=\int_\bbQ(\int_\bbP \frac{h}{\bbf}\circ H\cdot\bbg)\cdot k=
\iint_{\bbQP}\frac{h}{\bbf}\cdot(\bbg\cdot k)\circ H^{-1}=
\iint_{\bbQP}\frac{h}{\bbf}\cdot(\frac{k}{\bbf})\circ H^{-1}\cdot(\bbg\cdot\bbf)\circ H^{-1}=
\iint_{\bbQP}\frac{h}{\bbf}\cdot(\frac{k}{\bbf}\circ H^{-1})\cdot(\bbg\cdot\bbf)=
\int_{\bbQ}{h}(\int_\bbP\frac{k}{\bbf}\circ H^{-1}\cdot\bbg)=\scpr{h}{\cTs k}
$
\end{proof}

\subsection*{A sufficient condition for self-adjointness $\cT={\cTs}$}
Let $\tau$ be a measure preserving involution $\tau:\bbP\to\bbP$, $\tau\circ\tau={\rm id}$.
We can extend it to $\tau:\bbQP\to\bbQP$ by $\tau(q,p)=(q,\tau(p)$).
Assume that $\bbg$ is invariant with respect to $\tau$: $\bbg\circ\tau=\bbg$.
\begin{lem}\label{lem:sigma}
If $\tau\circ H^{-1}\circ\tau=H$ and $\bbg$ is invariant with respect to $\tau$ then ${\cTs}=\cT$.
\end{lem}
As an example take $\bbQ=\bbP=\Real$,
$\tau$ to   be the symmetry (reflection) of the space $\bbP$ with respect to the
0, $\tau(p)=-p$. An even $\bbg(p)=\bbg(-p)$ is invariant with respect to $\tau$.
The involution $\tau(p)=-p$ is applicable in the most common choice of $\bbg$: a centralized Gaussian distribution.
In a particular case of $\bbf$ also a Gaussian the Hamiltonian movement $H$ is a rotation and $H^{-1}$ an opposite rotation.
The spreading by $\bbg$ is symmetric and whatever mass is transported from $(q,p)$ to $H(q,p)=(Q,P)$
by $H$ the same mass will be transported from $(q,-p)$ to $H^{-1}(q,p)=(\bar{Q},\bar{P})=(Q,-P)$ by $H^{-1}$.
The projection onto $\bbQ$ will produce the same mass transported by both maps $\cT$ and $\cTs$. Clearly this extends to non-standard Gaussians.

\begin{proof}
Measure invariance means that $\int_\bbP a\circ\tau=\int_\bbP a$. Let $(\bar{Q},\bar{P})=H^{-1}(q,p)$ then
  $\tau\circ H^{-1}(q,p)=\tau(\bar{Q},\bar{P})=(\bar{Q},\tau(\bar{P})$ and
  $\cT h=\int_\bbP (h\cdot \bbg)\circ H=\int_\bbP (h\cdot \bbg)\circ \tau\circ H^{-1}\circ \tau=
  \int_\bbP (h\cdot \bbg)\circ \tau\circ H^{-1}=\int_\bbP(h\cdot\bbg)\tau(\bar{Q},\bar{P})=
  \int_\bbP (h(\bar{Q})\cdot \bbg(\tau(\bar{P}))=\int_\bbP h(\bar{Q})\bbg(\bar{P})=\int_\bbP h\circ H^{-1}\cdot\bbg\circ H^{-1}=\cTs h$.
\end{proof}

In what follows we shall assume that $\cT=\cTs$. If it is not the case we can  use the algorithm with  $\cS={\cTs}\circ\cT$ such that ${\cS^\dag}=\cS$.

\section{Limits of the sequences $\cT^n$ of self-adjoint operator}
From  $||\cT h||_\bbt<||h||_\bbt$ by induction we obtain $||\cT^nh||_\bbt<||h||_\bbt$ unless $h=\alpha\bbf$, when equality holds.
For $h\in\Lt$ let
\begin{equation}\label{eqndef: V}
V(h)=\inf ||T^n h||_\bbt^\bbt =\lim ||T^n h||_\bbt^\bbt\,.
\end{equation}
We see that for any $M$, $V(h)=V(\cT^M(h))$.
As we are interested in the limit of the sequence $\cT^n h$, for a given $h$ we can assume
that for an arbitrary $\epsilon>0$ we have $||h||_\bbt^\bbt<V+\epsilon$, taking a high iterate $\cT^M h$ instead of $h$ if needed.

By a corollary to Alaoglu Theorem bounded sets in reflexive $\Lt$ are weakly (the same as weakly*) compact. Any infinite sequence $\cT^n h$ have a weak  converging subsequence $\cT^{m_n}h\rightharpoonup h_\infty$,
meaning $\scpr{\cT^{m_n} h}{b}\to\scpr{h_\infty}{b}$ for every $b\in \Lt$.
\begin{lem}
  \label{lem:norm of hinf}
  If the operator $T$ is self adjoint then for any weak accumulation point
  $h_\infty$ of the sequence $T^n h$ we have
  \[
  \liminf ||T^nh||^2\le ||h_\infty||\limsup||T^n h||\,.
  \]
\end{lem}
\begin{proof}
With $ T^{m_n}h\rightharpoonup h_\infty$ let $N$ be such that the sequence  $(m_n-N)$ contains infinitely many even indices $2m$.
Then for $h'=T^N h$ we have $T^{2m}h'\rightharpoonup h_\infty$ as well, and:
\[
||T^m h'||^2=\scpr{T^m h'}{T^m h'}=\scpr{T^{2m} h'}{h'}\to\scpr{h_\infty}{h'}\le||h_\infty||\cdot||h'||\,.
\]
For an $\epsilon >0$ we can choose $N$ sufficiently large such that
$||T^m h'||=||T^{m+N}h|| \ge \liminf ||T^n h||-\epsilon$ and
$||h'||=||T^N h||\le \limsup_N ||T^N h||+\epsilon$.
Arbitrary choice of $\epsilon>0$ proves the Lemma.
\end{proof}
\begin{rem}
  For any bounded operator $T$ on $\Lt$ and any weak accumulation point $T^{m_n}h\rightharpoonup h_\infty $ we have
  \[
  ||h_\infty||\le\limsup ||T^{m_n}h||\le\limsup ||T^{n}h||\,.
  \]
\end{rem}
\begin{proof}
 This is standard: for an $\epsilon>0$ we can find $N$ such that for $m_n>N$ we have $||T^{m_n}h||\le\limsup||T^n h||+\epsilon$.
 Then $||h_\infty||_\bbt^\bbt=\scpr{h_\infty}{h_\infty}\leftarrow\scpr{T^{m_n}h}{h_\infty}\le ||T^{m_n}h||_\bbt||h_\infty||$.
\end{proof}
\begin{cor}{\label{cor:norm of hinf}}
  If for a self-adjoint operator $T$ the sequence of norms $||T^nh||$ converges then every weak accumulation point
  $T^{m_n}h\rightharpoonup h_\infty$ of the sequence $T^n h$
  has the norm $||h_\infty||= \lim ||T^n h||$ and is a strong limit of the same subsequence.
  \[
  ||T^{m_n}h-h_\infty||_\bbt^\bbt\to 0\,.
  \]
\end{cor}
\begin{proof}
  The value of norm $||h_\infty||$ follows from the previous Lemma and Remark. The strong convergence is standard again.
  Due to the strong convexity of the ball in $\Lt$ (or a direct manipulation of $||\cT^{m_n} h-h_\infty||_\bbt^\bbt$) a weak convergent sequence with the convergence of the norms to the norm of the limit converges also in the strong sense.
\end{proof}
\begin{prop}\label{prop:norm of hinf}
Assume $\cT=\cTs$. Then $\cT^n h $ converges strongly.
\[
||\cT^n h-\frac{\int h}{\int\bbf}\bbf||_\bbt^\bbt\to 0\,.
\]
\end{prop}
\begin{proof}
By Lemma~\ref{lem:cT prop norm} the sequence  $||\cT^n h||^2$ converges to $V(h)$ from~\eqref{eqndef: V}. By Corollary~\ref{cor:norm of hinf} every weak converging subsequence
of $\cT^n h$ converges strongly and the norm of the limit is equal $\sqrt{V(h)}$.
In particular if $\cT^{m_n}h\rightharpoonup h_\infty$ then $\cT^{m_n+1}h\rightharpoonup \cT(h_\infty)$ and $||h_\infty||_\bbt^\bbt=V(h)=||\cT h_\infty||_\bbt^\bbt$.
By Lemma~\ref{lem:cT prop}~\eqref{lem:cT prop eq} we have thus $h_\infty=\alpha \bbf=\cT h_\infty$, with $\alpha=\int h/\int\bbf$ by~\eqref{lem:cT prop int f}.
That means that every weak converging subsequence of $\cT^nh$ converges to the same limit $\alpha\bbf$. But any subsequence of $\cT^nh$ contains a weakly convergent subsequence, hence $\cT^n h$ converges weakly, and also strongly to $\alpha \bbf$.
\end{proof}
This concludes the proof of Theorem~\ref{thm:strong conv}.

\section{The operator $\cT$ as a kernel operator}
\subsection*{The kernel $K(q,Q)$}
In this section we shall assume a stronger version of  the covering property \eqref{ass:Q onto Q} of the map $H:\bbQP\to\bbQP$, $(Q,P)=H(q,p)$
on finite dimensional spaces $\bbQ$ and $\bbP$ (manifolds modeled on $\Real^d$, usually $\bbP$ will be a co-tangent space to $\bbQ$).
\begin{eqnarray}\label{eqn:def cQ q}
  \cQ_q:\bbP\to\bbQ\text{ defined by }\cQ_q(p)=Q(q,p)&&\text{ is a bijection for almost every }q\label{ass:bijection}\\
\label{eqn:def d cQ inv}  \frac{\partial Q(q,p)}{\partial p}&&\text{ exists and is invertible for a.e. }q\label{ass:Jacobian}
\end{eqnarray}
By assumption the function $\cQ_q^{-1}:\bbQ\to\bbP$, $p=\cQ_q^{-1}(Q)$ is well defined and so is
$\cP_q:\bbQ\to\bbP$, $\cP_q(Q)=P(q,p)=P(q,\cQ_q^{-1}(Q)$.
Define the Jacobian (determinant) $\cD_q$ of the partial derivative~\eqref{ass:Jacobian} by
\begin{equation}  \label{eqndef: D}
\cD_q(Q)=\left|{\rm det}\left(\frac{\partial \cQ_q(p)}{\partial p}\right)^{-1}\right|\in \Real, \text{ where } Q=\cQ_q(p)\,.
\end{equation}
The movement $(Q,P)=H(q,p)$ must be sufficiently smooth in order for $\cD_q$ to behave.
In order to simplify our reasoning  we shall assume that, similarly as in~\eqref{eqn:def cQ q} and~\eqref{eqn:def d cQ inv}, the map $\cP_p:\bbQ\to\bbP$, defined by $\cP_p(q)=P(q,p)$,  is a bijection with
invertible derivative of which the Jacobian (similarly as in~\eqref{eqndef: D}) $\cD_p(P)=|{\rm det}\,(\partial \cP_p(q)/\partial q)^{-1}|=|{\rm det}\,\partial q/\partial P|$.

\begin{rem}
  In such situation (one can assume that) the measure spaces $(\bbQ, dQ)$ and $(\bbP,dp)$ are isomorphic in the measurable sense. Then the Jacobian $\cD_q$ can be treated as the Radon-Nikodym derivative of the transport of the measure $(dp)$ on the fiber $\{q\}\times\bbP$ performed  by $H$ and the projection, in effect by $\cQ_q$, to the measure $(\cD_q\cdot dQ)$ on $\bbQ$. Similarly $\cD_p$ transfers the measure $dq$ on $\bbQ$ to the measure $\cD_p(P)\,dP$ on $\bbP$.
\end{rem}

With the change of variables by $p\mapsto Q=\cQ_q(p)$, $p=\cQ_q^{-1}(Q)$, $P=\cP_q(Q)$,
$dp= \cD_q(Q)\,dQ$ we get
$\cT h(q)=\int_\bbP h(Q)\bbg(P)\, dp=\int_\bbQ h(Q)\cdot\bbg(\cP_q(Q))\cD_q(Q)\,dQ$.
This shows that  the operator $\cT$ is a kernel operator, namely $\cT h(q)=\scpr{h(Q)}{\cK_q(Q)}$ where
$\cK_q(Q)=K(q,Q)$ is defined by
  \begin{eqnarray}\label{eqn:def K}
    K(q,Q)&=&\bbf(Q)\cdot\bbg(\cP_q(Q))\cdot\cD_q(Q)\qquad       \text{ as then} \label{eqn:ct kernel}\\
    \cT h(q)&=&\int_\bbP h(Q(q,p))\bbg(P(q,p))\,dp=\int_\bbQ \frac{h(Q)\cdot\bbg(\cP_q(Q))\bbf(Q)}{\bbf(Q)}\cD_q(Q)\,dQ
    \nonumber\\
    &=&\int_\bbQ\frac{h(Q)\cdot K(q,Q)}{\bbf(Q)}\,dQ
    =\scpr{h}{\cK_q}
    \nonumber
  \end{eqnarray}
 When a kernel operator has a finite $\Lt$ norm it is compact, and then its spectrum is discrete except the unique possible accumulation point at 0.
  We know by Remark~\ref{rem:cT spec} that the spectrum of $\cT$ has all the eigenvalues inside the unit disk except for the eigenvalue 1 which has multiplicity 1. The compactness of $\cT$ would provide the spectral gap and exponential convergence in norm of each sequence $\cT^n h$ to its fixed point limit $\alpha\bbf$ with the rate of convergence given by the second largest eigenvalue, in this case, strictly smaller than 1.
The norm for $A(q,Q)\in\Lt(\bbQQ)$ is given by
  \[
 ||A||_2^2=\iint_\bbQQ \frac{A^2(q,Q)}{\bbf(q)\cdot\bbf(Q)}\,dQ\,dq\,.
  \]
  For a.e. $q$ the maps $Q\mapsto \cA_q(Q)=A(q,Q)$ should belong to $\Lt(\bbQ)$ (with respect to $Q$, thus $\bbf(Q)$ in the denominator) and then the map
  defined by the norms  $q\mapsto ||\cA_q||_2$ should belong to $\Lt(\bbQ)$ (with respect to $q$ thus $\bbf(q)$ in the denominator),
  in both cases use the norm defined in the first line of Section~\ref{sec: operator cT}. One might be more comfortable working  with an analogous kernel expression in the space $\tilde{L}^2$.

\begin{lem}[The norm of the kernel $K$]\label{lem:norm of K}
The $\Lt$ norm of the kernel $K$~\eqref{eqn:def K} can be expressed as:
\[
||K||_\bbt^\bbt=\iint_\bbPP \bbg(p)\bbg(P) \cD_q(Q)\cD_p(P)\,dP\,dp\,,
\]
where $Q$ and $q$ are well defined functions of $p$ and $P$, $q=\cP_p^{-1}(P)$ and $Q=\cQ_q(p)$.
\end{lem}
\begin{proof}
Using invariance (with $p=\cQ_q^{-1}(Q)$, $P=\cP_q(Q)$) we have:
\[
K^2(Q,q)=\bbf^2(Q)\bbg^2(P)\cD_q^2(Q)=
\bbf(Q)\bbf(Q)\bbg(P)\bbg(P)\cD_q^2(Q)=
\bbf(Q)\bbf(q)\bbg(p)\bbg(P)\cD_q^2(Q)\,.
\]
In the integral expression for the norm we change both variables $q$ and $Q$ to variables $p$ and $P$ using
the Jacobians $\cD_q(Q)\,dQ=dp$ and $dq=\cD_p(P)\,dP$
\begin{align}\label{eqn: norm of K}
||K||_\bbt^\bbt&=
\iint_\bbQQ\frac{K^2(q,Q)}{\bbf(Q)\bbf(q)}\,dQ\,dq=
\iint_\bbQQ \bbg(p)\bbg(P)\cD_q^2(Q)\,dQ\,dq\\
&=\iint_\bbQP  \bbg(p)\bbg(P)\cD_q(Q)\,dP\,dq=
\iint_\bbPP  \bbg(p)\bbg(P)\cD_q(Q)\cD_p(P)\,dP\,dp
\end{align}
\end{proof}
\begin{cor}
  \label{cor:K compact if det bounded}
  If the product of determinants $\cD_q(Q)\cdot\cD_p(P)$ is uniformly bounded from above then the $\Lt$ norm of the kernel $K$ is finite and the operator $\cT$ is compact.
\end{cor}
  \subsection*{Hamiltonian movement}

In the following we shall prove that $||K(q,Q)||_2^2$ is finite in a special case of a Hamiltonian movement when $\bbQ=\bbP=\Real^d$.
The distribution of choice $\bbg(p)$ is usually the standard Gaussian
(mean 0 and covariance equal to identity matrix, that is $-\log\bbg(p)=\scpr{p}{p}/2$ up to the irrelevant additive normalising constant), but the proof is provided for any uniformly strictly concave distribution.
Define $\bbU(Q)=-\log(\bbf(Q))$ then $\bbU''(Q)=\partial^2 (-\log \bbf(Q))/\partial Q^2$, a symmetric matrix by assumed continuous differentiability.
We assume that the target distribution $\bbf$ is  \emph{uniformly strictly logarithmic concave},
that is $\bbU''$ is a (strictly) positive operator, \emph{i.e.} it is  bounded away (in either norm or spectrum sense) from $0$ and $\infty$ uniformly on $q\in\bbQ$.
Then $\bbU''(Q)$ can be bounded away from $0$ and $\infty$ by two symmetric, strictly positive bounded operators constant with respect to $Q$, then the distribution $\bbf$ can be then estimated both from above and below by two Gaussians (with some positive finite multiplicative constants). Similar statements hold for $\bbg$ and $\bbV(P)=-\log(\bbg(P))$.

Consider the spaces $\LQt$, $\LPt$ and $\LQPt$ of functions on $\bbQ$ (positions), $\bbP$ (momenta) and
$\bbQP$ (configurations) to $\Real$ with the appropriate (integral) norms.
Given $0\le \bbf\in \LQt$ and $0\le \bbg\in\LPt$ we can define
(\emph{potential energy}) $\bbU:\bbQ\to\Real$, $\bbU(q)=-\log(\bbf(q))$ and (\emph{kinetic energy})
$\bbV:\bbP\to\Real$, $\bbV(p)=-\log(\bbg(p))$.
We see that if $\bbf(q)$ and $\bbg(p)$ represent the densities of probability distributions of two independent variables then
$\bbf(q)\cdot\bbg(p)=\exp(-(\bbU(q)+\bbV(p))$ represents a density of their (independent) joint distribution.
We shall use the name \emph{Hamiltonian} for the total energy $\cH=\bbU+\bbV$.
\begin{rem}
In fact one may consider a more general case when $\bbg=\bbg(q,p)$ (understood as conditional $\bbg(p|q)$) and thus $\bbV=\bbV(p|q)$. However to simplify the calculations we shall deal only with the case of $\bbg$ independent on~$q$.
\end{rem}
The Hamiltonian energy provides the following (Hamiltonian) dynamics $(q,p)\mapsto(Q,P)$, where $(Q,P)=(Q_t(q,p),P_t(q,p))$ is the position after time $t$ of the point starting at $(q,p)$ ruled by the system of equations:
\begin{eqnarray}\label{eqn: d QP dt}
  \dot{Q}=\frac{dQ}{dt}&=&\frac{\partial \cH(Q,P)}{\partial P}\\
\nonumber  \dot{P}=\frac{dQ}{dt}&=-&\frac{\partial \cH(Q,P)}{\partial Q}\,.
 \end{eqnarray}
 The dot derivative is the derivative with respect to time $\dot{A}=\partial A/\partial t$. We see that the normalizing constants of $\bbg$ and $\bbf$ are irrelevant to the motion.
 Formally the solutions can be written as
 \begin{eqnarray}\label{eqn: int QP dt}
   Q(t) &=& q+\int_0^t \frac{\partial \cH}{\partial P}(Q(s), P(s))\,ds=q+\int_0^t \bbV'(P(s))\,ds \\
 \nonumber  P(t) &=& p-\int_0^t \frac{\partial \cH}{\partial Q}(Q(s), P(s))\,ds= p-\int_0^t \bbU'(Q(s))\,ds \,,
 \end{eqnarray}
 assuming all functions are sufficiently regular.

 The movement $H$ is defined by $(q,p)\mapsto H_t(q,p)=(Q_t,P_t)$, and the  map $\cT_t$ is defined by~\eqref{eqdef:cT} using $H_t$.
 For a given fixed time $t$ we shall  skip the subscript $t$. For a function $W\in\LQP$  we denote $t\mapsto W_t(q,p)=(W\circ H_t)(q,p)=W(Q,P)$.

The value of the Hamiltonian $\cH$ does not change along the trajectories and
the Hamiltonian motion conserves the Lebesgue measure:
 \begin{eqnarray}
  && \cH\circ H=\cH, \qquad\text{or after taking the exponent}\qquad (\bbf\cdot \bbg)\circ H=\bbf\cdot \bbg\,.
\label{eqn:constant Hamiltonian Rd} \\
   &&\iint_{\bbQP} W_t\,d(q p)=\iint_{\bbQP} W\, d(q p) \text{ for any }W\in\LQP\,,
\label{eqn:Hamiltonian of Lebesgue Rd}
 \end{eqnarray}
 which corresponds to the properties~\eqref{ass:inv fg} and~\eqref{ass:inv A}.

 We define $\cT$ as in ~\eqref{eqdef:cT}, to stress the dependence of $\cT$ on the choice of $t$ we use the notation $\cT_t$.
Given $t$ and $n$ the iterate of the map $\cT^n_t$ is in general different from the map $\cT_{tn}$ with time $nt$.

The subset of (target) functions $\bbf$ with some interest has usually some additional properties: $\bbU,\bbV\to +\infty$, as $|q|,|p|\to\infty$, fast enough, so that the functions $\bbf, \bbg$ are bounded, integrable and vanish at infinity (or at boundaries of the support) meaning that  $\cH$ escapes to infinity when $(q,p)$ approaches these boundaries.
 This assures that the level sets of the Hamiltonian and therefore the  trajectories are bounded (and closed).
 Often additionally the derivatives of $\bbU$ and $\bbV$ (or $\bbf$ and $\bbg$) are zero at a unique point (no stationary points except this one).

\subsection*{Differential equation solutions to the gradient of the Hamiltonian flow}
Assume that all the functions involved have sufficient smoothness, so that derivative exists and their order can be changed.
In order to estimate the norm of the kernel we need to have a good control on the Jacobians $\cD_q(Q)$ and $\cD_p(P)$, in other words the
partial derivatives of the time evolution of configuration with respect to initial configuration
$\partial (Q,P)/\partial (q,p)$.

By assumption that $\cH(Q,P)=\bbU(Q)+\bbV(P)$ (\emph{i.e.} the spreading $\bbg$ does not depend on the position $q$) its mixed second derivatives $\partial^2 \cH/\partial Q\partial P=0$  vanish.

\begin{lem}[Evolution of the dependence on the initial configuration]\label{lem:evol d QP d qp}
  Under the assumption that  $\cH(Q,P)=\bbU(Q)+\bbV(P)$
  the  derivative of the motion $(Q,P)$ with respect to the starting configuration $(q,p)$ satisfy the following time evolution equation:
  \begin{equation}\label{eqn: ODE dQP dqp}
  \frac{\partial}{\partial t}
  \left(
    \begin{array}{cc}
     \frac{\partial Q}{\partial q} & \frac{\partial Q}{\partial p} \\
     \frac{\partial P}{\partial q} & \frac{\partial P}{\partial p}
    \end{array}
  \right)
  =
  \left(
  \begin{array}{cc}
    0& \bbV'' \\
    -\bbU''  & 0
  \end{array}
  \right)
  \cdot
  \left(
    \begin{array}{cc}
     \frac{\partial Q}{\partial q} & \frac{\partial Q}{\partial p} \\
     \frac{\partial P}{\partial q} & \frac{\partial P}{\partial p}
    \end{array}
  \right);
\qquad
  \left(
    \begin{array}{cc}
     \frac{\partial Q}{\partial q} & \frac{\partial Q}{\partial p} \\
     \frac{\partial P}{\partial q} & \frac{\partial P}{\partial p}
    \end{array}
  \right)_{t=0}=
 \left(
  \begin{array}{cc}
    I & 0 \\
    0 & I
  \end{array}
  \right)
 \end{equation}
where $\bbU''=\bbU''(Q)={\partial^2 \cH}/{\partial Q^2}$ and
$\bbV''=\bbV''(P)={\partial^2 \cH}/{\partial P^2}$.

\end{lem}
\begin{proof}
The initial condition $\partial (Q,P)/\partial (q,p)_{t=0}=I$ can be calculated from~\eqref{eqn: int QP dt}, assuming sufficient continuity one can change the order of derivative and integration.
Then as $t\to 0$ only the first term becomes relevant.
As for the equation, we first calculate the derivatives with respect to $q$:
\[
\frac{\partial}{\partial q}\frac{\partial \cH}{\partial P} =
\frac{\partial^2\cH}{\partial Q \partial P}\cdot\frac{\partial Q}{\partial q} +\frac{\partial^2\cH}{\partial P^2}\cdot\frac{\partial P}{\partial q} =
0\cdot\frac{\partial Q}{\partial q} +\bbV''\cdot\frac{\partial P}{\partial q}
=\bbV''\frac{\partial P}{\partial q}\,.
\]
Similarly $\partial^2 \cH/\partial q\partial Q=\bbU''\cdot\partial Q/\partial q$. It is clear that the calculation holds after exchanging every $q$ by $p$.
Now we differentiate the Hamiltonian equations~\eqref{eqn: d QP dt} with respect to initial configuration $(q,p)$ and change the order of differentiation, for example
\[
\frac{\partial}{\partial t}\left(\frac{\partial Q}{\partial q}\right)=
\frac{\partial}{\partial q}\left(\frac{\partial Q}{\partial t}\right)=
\frac{\partial }{\partial q}\frac{\partial \cH} {\partial P}
=\bbV''\frac{\partial P}{\partial q}\,,
\]
and again similarly $\partial^2 P /\partial t \partial q =-\bbU''\cdot\partial Q/\partial q$. The calculation holds when exchanging $q$ for $p$.
\end{proof}
Before we proceed with the proof we remind that for $U,V$ symmetric, positive definite operators on $\Lt$ their symmetric positive definite square roots are uniquely defined. For example, for $V<I$ (that is $I-V$ is positive definite, which can be achieved by a normalization trick) $\sqrt{V}=I-R$ where $R$ is a limit of the (strongly converging) sequence $R_{n+1}=(I-(V-R_n^2))/2$, $R_0=0$. Also $VU$ and $UV$ are positive definite (but not necessarily symmetric, when non commuting), as for example $VU=\sqrt{U}^{-1}(\sqrt{U}V\sqrt{U})\sqrt{U}$ is similar via a symmetric operator $\sqrt{U}$ to a symmetric positive definite $\sqrt{U}V\sqrt{U}$. Using this we can define $\sqrt{VU}=\sqrt{U}^{-1}\sqrt{\sqrt{U}V\sqrt{U}}\sqrt{U}$, and similarly $\sqrt{UV}$.

Below the functions are defined by their power series
$\exp(x)=\sum_{n=0}^\infty x^n/n!$, $\sin(x)=\sum_{n=0}^\infty (-1)^n x^{2n+1}/(2n+1)!$,
  $\sinc(x)=x^{-1}\sin(x)=\sum_{n=0}^\infty (-1)^n x^{2n}/(2n+1)!$, which is well defined even when $x^{-1}$ is not
  and $\cos(x)=\sum_{n=0}^\infty (-1)^n x^{2n}/{(2n)!}$.

\begin{lem}[Exponential function of a matrix]
  \label{lem:exp B}
Let  $V,U$ be  symmetric, positive definite linear operators in $\Lt$.
If for $t\in\Real$
\[
\cC=\left(
  \begin{array}{cc}
    0 & tV \\
    -tU & 0
  \end{array}
  \right)
\]
then for $A=\sqrt{VU}$ and $B=\sqrt{UV}$ we have
\[
  \exp(\cC)=\sum_{n=0}^\infty(-1)^n\cC^n=
  \left(
   \begin{array}{cc}
    \cos(tA) & tV\,\sinc (tB) \\
    -tU\,\sinc(tA) & \cos(tB)
  \end{array}
  \right)\,.
    \]
\end{lem}
\begin{proof}

  From direct calculation of $\cC^2$ we have the following powers of $\cC$:
 \begin{align*}
\cC^{2n}&=(-1)^n \left(\begin{array}{cc}
    (VU)^nt^{2n}& 0 \\
    0 & (UV)^nt^{2n}
  \end{array}\right)=
(-1)^n \left(\begin{array}{cc}
    (At)^{2n}& 0 \\
    0 & (Bt)^{2n}
  \end{array}\right);
\\
\cC^{2n+1}&=(-1)^n \left(\begin{array}{cc}
    0&V(UV)^nt^{2n+1} \\
    -U(VU)^{n}t^{2n+1} & 0
  \end{array}\right)\\
  &=
 (-1)^n \left(\begin{array}{cc}
    0&tV(Bt)^{-1}(Bt)^{2n+1} \\
    -tU(At)^{-1}(At)^{2n+1} & 0
  \end{array}\right)\,\\
\end{align*}
so that:
\begin{align*}
  \exp(\cC)=\sum_{n=0}^\infty(-1)^n&
  \left(
  \begin{array}{cc}
    \frac{(tA)^{2n}}{(2n)!} & tV(tB)^{-1}\frac{(tB)^{2n+1}}{(2n+1)!}  \\
    -tU(tA)^{-1}\frac{(tA)^{2n+1}}{(2n+1)!}  & \frac{(tB)^{2n}}{(2n)!}
  \end{array}
  \right)
   \\
   &
   =
  \left(
   \begin{array}{cc}
    \cos(tA) & tV\,\sinc(tB) \\
    -tU\,\sinc(tA) & \cos(tB)
  \end{array}
  \right)
   \,.
\end{align*}
We need to use the additional $U$ and $V$ to compensate for odd powers on the off-diagonal.
\end{proof}
In the following Proposition let $\cH(Q,P)=\bbU(Q)+\bbV(P)$, and
\begin{align}\label{eqndef:U}
  U(t)&=\frac{1}{t}\int_0^t \frac{\partial^2 \cH(Q(s),P(s))}{\partial Q^2}\,ds=\frac{1}{t}\int_0^t \bbU''(Q(s))\,ds\\
  \label{eqndef:V}
  V(t)&=\frac{1}{t}\int_0^t \frac{\partial^2 \cH(Q(s),P(s))}{\partial P^2}\,ds=\frac{1}{t}\int_0^t \bbV''(P(s))\,ds\,.
\end{align}
\begin{prop}[Solution of the evolution of the dependence on initial conditions]
  \label{prop:d QP d qp}
  Assume that the target $\bbf$ and auxiliary $\bbg$  distributions are both  strictly log-concave. Then $\bbU''(s)$ and $\bbV''$ are symmetric positive definite, and so are $U(t)$ and $V(t)$. Let $A=A(t)=\sqrt{V(t)U(t)}$ and $B=B(t)=\sqrt{U(t)V(t)}$.
Then the solution of the evolution equation in Lemma~\ref{lem:evol d QP d qp}~\eqref{eqn: ODE dQP dqp} is given by:
\[
 \left(
    \begin{array}{cc}
     \frac{\partial Q}{\partial q} & \frac{\partial Q}{\partial p} \\
     \frac{\partial P}{\partial q} & \frac{\partial P}{\partial p}
    \end{array}
  \right)(t)=
  \left(
   \begin{array}{cc}
    \cos(t A) & t V(t)\,\sinc (t B) \\
    -tU(t)\,\sinc(t A) & \cos(t B)
  \end{array}
  \right)\,.
\]
\end{prop}
\begin{proof}
  The solution to a linear differential equation $\dot{X}=A(t) X$ is equal to $X(t)=\exp(\int_{s=0}^t A(s)\,ds)\cdot X(0)$
  and we use Lemma~\ref{lem:exp B}.
\end{proof}
\subsection*{Proof of Theorem~\ref{thm:Hamil  geom conv}}\label{subsec:Proof Thm geom conv}
\begin{proof}
It is enough to prove that  the Jacobians $\cD_q(Q)$ and $\cD_p(P)$ are uniformly bounded away from $0$ and $\infty$, and we can use Corollary~\ref{cor:K compact if det bounded} which says that
  the kernel $K$ defined in~\eqref{eqn:def K} is then bounded in $\Lt$ which makes the operator $\cT$ compact, which yields to the spectral gap. Then, as the (maximal) eigenvalue $1$ has multiplicity
  $1$, the spectrum $\sigma$ of the operator $\cT_N$ on the closed hyperplane $N=\{h:\int h=0\}$
  orthogonal to the eigen direction of fixed points $\{\bbf\cdot\Real\}$, which is contained inside the open unit disk
  $\spectrum(\cT_N)\subset \{\mu\in \Complex:|\mu|<1\}$ (in fact in the interval $[0,1)$ as $\cT$ is positive and in our case selfadjoint). The spectrum has only 0 as the possible accumulation point. We get $\sup |\spectrum(\cT_N)|<1$, the radius is  in fact the second largest eigenvalue, which secures the geometrical convergence to 0 on $N$ and to $\bbf\cdot(\int h/\int\bbf)$ in $\Lt$.

We deal with  finite dimensional $\bbQ$ and $\bbP$, both modeled by $\Real^d$. The operators  $\partial Q/\partial p $ and $\partial P/\partial q$
are both symmetric $d\times d$ matrices with the determinant equal to the product of their real eigenvalues.
For a (positive symmetric) matrix $M$ and a (positive) function $\phi$ defined by the power series the (real positive) eigenvalues of the (positive symmetric) matrix $\phi(M)$ are equal to the images under $\phi$ of the (real positive) eigenvalues of $M$.

By Proposition~\ref{prop:d QP d qp}, $\partial Q/\partial p=tV\sinc(t\sqrt{UV})$, where $U=U(t)$ and $V=V(t)$ were defined in~\eqref{eqndef:U} and~\eqref{eqndef:V}.
By the assumption on uniform strict concavity the spectra $\spectrum(\bbU'')$ and $\spectrum(\bbV'')$ are  uniformly bounded away from 0 and infinity by  $0<\lambda=\inf(\spectrum(\bbU''),\spectrum(\bbV''))\le \sup(\spectrum(\bbU''),\spectrum(\bbV''))=\Lambda<\infty$ so do the spectra of running averages $U$ and $V$ and also the spectra of $A=\sqrt{VU}$ and $B=\sqrt{UV}$ as in ~Proposition~\ref{prop:d QP d qp}.
Consequently, for small $t$, such that $0<t\Lambda<\pi/2$,  $\spectrum(\partial Q)/\partial p)$ is bounded,
from above by $t\Lambda$ and from below by $t\sinc(t\lambda)$.
Similarly $\spectrum(\partial P/\partial q)$ is bounded from below by $-t\Lambda$ and from above by $-t\sinc(t\lambda)$.
All these bounds are uniformly away from $0$ and $\infty$. Finally for the product of determinants $\cD_q(Q)\cdot\cD_p(P)$ defined by~\eqref{eqndef: D} have uniform bounds away from $0$ and $\infty$ by $(t\,\sinc(t\lambda))^{-2d}$ from above and by $ (t\Lambda)^{-2d}$ from below.
 \end{proof}
\bibliographystyle{abbrv}

\end{document}